\documentclass[a4paper,11pt]{article}
\usepackage{fullpage}
\usepackage{amsthm, amsfonts,mathrsfs,amssymb,amsmath,clrscode,multirow,graphicx,stmaryrd,mathtools}
\usepackage{graphicx}
\usepackage{mathptmx}
\makeatletter
\newcommand\incircbin{ \mathpalette\@incircbin}
\newcommand\@incircbin[2]{\mathbin{\ooalign{\hidewidth$#1#2$\hidewidth\crcr$#1\bigcirc$}}}
\makeatother
\newcommand{\smallerprod}{\leftslice}
\newcommand{\domprod}{\ast}
\newcommand{\Boolprod}{\cdot}

\newcommand{\ket}[1]{\vert #1 \rangle}

\newcommand{\Int}{\mathbb Z}

\newcommand{\ceil}[1]{\left\lceil #1 \right\rceil}

\newcommand{\mybar}[1]{\lambda}

\newcommand{\arrayf}[1]{\mathsf{#1}}

\newcommand{\col}[2]{\mathrm{col}(#1,#2)}

\newcommand{\poly}{\mathrm{poly}}

\newcommand{\polylog}{\mathrm{polylog}}

\newtheorem{theorem}{Theorem}[section]

\newtheorem{proposition}{Proposition}[section]
\newtheorem{definition}{Definition}[section]
\newtheorem{lemma}{Lemma}[section]
\newtheorem{fact}{Fact}
\begin{document}
\pagestyle{plain}
\author{%
Fran{\c c}ois Le Gall \\ The University of Tokyo\\ \texttt{legall@is.s.u-tokyo.ac.jp}
\and Harumichi Nishimura \\ Nagoya University\\ \texttt{hnishimura@is.nagoya-u.ac.jp} 
}
\title{Quantum Algorithms for Matrix Products over Semirings}
\date{}

\maketitle
\begin{abstract}
In this paper we construct quantum algorithms for matrix products over several algebraic structures called semirings, including the $(\max,\min)$-matrix product, the distance matrix product and the Boolean matrix product. In particular, we obtain the following results.
\begin{itemize}
\item
We construct a quantum algorithm computing the product of two $n\times n$ matrices over the $(\max,\min)$ semiring with time complexity $O(n^{2.473})$. In comparison, the best known classical algorithm for the same problem, by Duan and Pettie (SODA'09), has complexity $O(n^{2.687})$. As an application, we obtain a $O(n^{2.473})$-time quantum algorithm for computing the all-pairs bottleneck paths of a graph with~$n$ vertices, while classically the best upper bound for this task is $O(n^{2.687})$, again by Duan and Pettie.
\item
We construct a quantum algorithm computing the $\ell$ most significant bits of each entry of the distance product of two $n\times n$ matrices in time $O(2^{0.64\ell} n^{2.46})$. In comparison, prior to the present work, the best known classical algorithm for the same problem, by Vassilevska and Williams (STOC'06) and Yuster (SODA'09), had complexity $O(2^{\ell}n^{2.69})$. Our techniques lead to further improvements for classical algorithms as well, reducing the classical complexity to $O(2^{0.96\ell}n^{2.69})$, which gives a sublinear dependency on $2^\ell$.
\end{itemize}
The above two algorithms are the first quantum algorithms that perform better than the $\tilde O(n^{5/2})$-time straightforward quantum algorithm based on quantum search for matrix multiplication over these semirings. We also consider the Boolean semiring, and construct a quantum algorithm computing the product of two $n\times n$ Boolean matrices that outperforms the best known classical algorithms for sparse matrices. For instance, if the input matrices have $O(n^{1.686\ldots})$ non-zero entries, then our algorithm has time complexity $O(n^{2.277})$, while the best classical algorithm has complexity $\tilde O(n^{\omega})$, where $\omega$ is the exponent of matrix multiplication over a field (the best known upper bound on $\omega$ is $\omega<2.373$). 
\end{abstract}


\setcounter{footnote}{0}
\section{Introduction}

\noindent{\bf Background.}
Matrix multiplication over semirings has a multitude of applications
in computer science, and in particular in the area of graph algorithms (e.g., \cite{Duan+SODA09,Shapira+SODA07,VassilevskaPhD08,Vassilevska+STOC06,Vassilevska+09,YusterSODA09}). 
One example is Boolean matrix multiplication, related for instance to the computation
of the transitive closure of a graph, where the product  
of two $n\times n$ Boolean matrices $A$ and $B$ is
defined as the $n\times n$ Boolean 
matrix $C=A\Boolprod B$ such that $C[i,j]=1$ if and only if there exists 
a $k\in\{1,\ldots,n\}$ such that $A[i,k]=B[k,j]=1$. 

More generally, given a set $R\subseteq \Int\cup\{-\infty,\infty\}$ and two binary operations 
$\oplus\colon R\times R\to R$ and $\odot\colon R\times R\to R$,
the structure $(R, \oplus,\odot)$ is a semiring if it behaves like a ring 
except that there is no requirement on the existence of an 
inverse with respect to the operation~$\oplus$.
Given two $n\times n$ matrices $A$ and $B$ over~$R$, the matrix product over
$(R,\oplus,\odot)$ is the $n\times n$ matrix $C$ defined as 
$
C[i,j]=\bigoplus_{k=1}^n \left(A[i,k]\odot B[k,j]\right)
$
for any $(i,j)\in\{1,\ldots,n\}\times \{1,\ldots,n\}$.
The Boolean matrix product is simply the matrix product over the semiring $(\{0,1\},\vee,\land)$.
The $(\max,\min)$-product and the distance product,
which both have applications to a multitude of tasks in graph theory such as constructing
fast algorithms for  
all-pairs paths problems
(see, e.g., \cite{VassilevskaPhD08}), are the matrix products over the 
semiring $(\Int\cup\{-\infty,\infty\},\max,\min)$
and the semiring $(\Int\cup\{\infty\},\min,+)$, respectively.

Whenever the operation $\oplus$ is such that
a term as $\bigoplus_{k=1}^n x_k$ can be computed in $\tilde O(\sqrt{n})$ time
using quantum techniques  
(e.g., for $\oplus=\vee$ using Grover's algorithm~\cite{GroverSTOC96} or for $\oplus=\min$ and $\oplus=\max$
using quantum algorithms for minimum finding~\cite{Durr+96}) and each operation $\odot$ can be implemented in 
$\polylog(n)$ time, the product of two $n\times n$ matrices over the semiring $(R,\oplus,\odot)$ can 
be computed in time $\tilde O(n^{5/2})$ on a quantum computer.\footnote{In this paper the notation 
$\tilde O(\cdot)$ suppresses the $n^{o(1)}$ factors.} This is true for instance 
for the Boolean matrix product, and 
for both the $(\max,\min)$ and distance matrix products.

A fundamental question is whether we can do better than those $\tilde O(n^{5/2})$-time straightforward quantum  
algorithms. For the Boolean matrix product, the answer is affirmative since it can be computed classically in time 
$\tilde O(n^\omega)$, where $\omega<2.373$ is the exponent of square matrix multiplication over a field. 
However, Boolean matrix product appears to be an exception, and for most 
semirings it is not known if matrix multiplication can be done 
in $\tilde O(n^\omega)$-time.
For instance, the best known classical algorithm for the $(\max,\min)$-product, by Duan and Pettie \cite{Duan+SODA09}, has time complexity $\tilde O(n^{(3+\omega)/2})=O(n^{2.687})$
while, for the distance product,  
no truly subcubic classical algorithm is even known (without introducing assumptions on the matrices). \vspace{2mm}

\noindent{\bf Our results.} 
We construct in this paper the first quantum algorithms 
with exponent strictly smaller than $5/2$ for matrix multiplication over several semirings.

We first obtain the following result for matrix multiplication over the $(\max,\min)$ semiring.
\begin{theorem}\label{the2}
There exists a quantum algorithm that computes,
with high probability, 
the 
$(\max,\min)$-product of two $n\times n$ matrices in time 
$O(n^{2.473})$.
\end{theorem}
\noindent
In comparison, the best known classical algorithm for the $(\max,\min)$-product, by Duan and Pettie~\cite{Duan+SODA09}, has time complexity $\tilde O(n^{(3+\omega)/2})=O(n^{2.687})$,
as mentioned above.
The $(\max,\min)$-product has mainly been studied in the field
in fuzzy logic~\cite{Dubois+80} under the name \emph{composition of relations}
and in the context of computing the all-pairs bottleneck paths of a graph (i.e., 
computing, for all pairs $(s,t)$ of vertices in a graph,
the maximum flow that can be routed between $s$ and~$t$). 
More precisely, it is well known (see, e.g., \cite{Duan+SODA09,Shapira+SODA07,Vassilevska+09}) that
if the $(\max,\min)$-product of two $n\times n$ matrices
can be computed in time $T(n)$,
then the all-pairs bottleneck paths of a graph with $n$ vertices can 
be computed in time $\tilde O(T(n))$. 
As an application of Theorem~\ref{the2}, 
we thus obtain a $O(n^{2.473})$-time quantum algorithm computing the all-pairs 
bottleneck paths of a graph of $n$ vertices, while classically the best upper bound for this task is $O(n^{2.687})$,
again from~\cite{Duan+SODA09}.

In order to prove Theorem \ref{the2}, we construct 
a quantum algorithm that computes the  
product of two $n\times n$ matrices over the existence dominance semiring (defined
in the next section) in time $\tilde O(n^{(5+\omega)/3})\le O(n^{2.458})$. 
The dominance product has applications in computational geometry \cite{Matousek91} and graph
algorithms \cite{Vassilevska+STOC06}
and, in comparison, the best known classical algorithm for this product \cite{YusterSODA09}
has complexity $O(n^{2.684})$.
Computing efficiently the existence dominance product is, nevertheless, 
not enough for our purpose.
We introduce (in Section \ref{sec_existence})
a new generalization of it that we call
the \emph{generalized existence dominance product},
and develop both quantum and classical algorithms that compute efficiently this product.
This is the most technical part of this paper.

We also show (in Subsection \ref{subsec_dist}) how 
these results for the generalized existence dominance product can 
be used to construct classical and quantum
algorithms computing the $\ell$ most significant bits of 
each entry of the distance product of two $n\times n$ matrices.
In the quantum setting, we obtain time complexity
$
\tilde O\left(2^{0.640\ell} n^{(5+\omega)/3}\right)\le O(2^{0.640\ell} n^{2.458}).
$
In comparison, prior to the present work, the best known classical algorithm for the same problem by Vassilevska and Williams~\cite{Vassilevska+STOC06} had complexity 
$
\tilde O\big( 2^{\ell}n^{(3+\omega)/2}\big)\le O(2^{\ell} n^{2.687}),
$
with a slight improvement on the exponent of $n$ obtained later by Yuster~\cite{YusterSODA09}.
We obtain an improvement for this classical time complexity as well,
reducing it 
to $\tilde O\big( 2^{0.960\ell}n^{(3+\omega)/2}\big)$, which gives a sublinear dependency on $2^\ell$.

These results are, to the best of our knowledge, the first quantum algorithms for matrix 
multiplication over semirings other than the Boolean semiring improving over the 
straightforward $\tilde O(n^{5/2})$-time quantum algorithm, and the first nontrivial quantum 
algorithms offering a speedup with respect to the best classical algorithms 
for matrix multiplication when no assumptions are made on the sparsity of the matrices involved 
(sparse matrix multiplication is discussed below).
This shows that, while quantum algorithms may not be able to outperform the classical 
$\tilde O(n^\omega)$-time algorithm for matrix multiplication of (dense) matrices over a ring,
they can offer a speedup for matrix multiplication over other algebraic structures.

We finally investigate under which conditions quantum algorithms faster than the best known 
classical algorithms can be constructed for Boolean matrix multiplication.
This question has been recently studied extensively in the output-sensitive scenario \cite{Buhrman+SODA06,Jeffery+ICALP12,LeGallSODA12,LeGallISAAC12}, for which quantum algorithms multiplying two $n\times n$ Boolean matrices with query complexity $\tilde O(n\sqrt{\lambda})$ and time complexity
$\tilde O(n\sqrt{\lambda}+\lambda\sqrt{n})$ were constructed, where $\lambda$ denotes the number of non-zero entries in the output matrix.
In this work, we focus on the case where the input matrices are sparse (but not necessarily the output matrix), 
and evaluate the performance of quantum algorithms in this scenario. Our result (Theorem \ref{prop}) 
shows how several standard combinatorial ideas for sparse Boolean matrix multiplication can be adapted in the quantum
setting, and used to construct quantum algorithms faster than the best known classical algorithms. In particular, we obtain the following result. 

\begin{theorem}[simplified version]\label{th_Booleansparse}
Let $A$ and $B$ be two $n\times n$ Boolean matrices each containing at most $m$ non-zero entries.
There exists a quantum algorithm that computes, with high probability,
the Boolean matrix product $A\Boolprod B$ and has time complexity
\[
\left\{\begin{array}{ll}
\tilde O(n^2)&\textrm{if }m\le n^{1.151},\\
\tilde O\left(m^{0.517}n^{1.406}\right)&
\textrm{if }n^{1.151}\le m\le n^{\omega-1/2},\\
\tilde O(n^\omega)&\textrm{if }n^{\omega-1/2}\le m\le n^2.\\
\end{array}
\right. 
\] 
\end{theorem}
\noindent
The complexity of the algorithm of Theorem \ref{th_Booleansparse} is the piece-linear function of $\log_n(m)$
represented in Figure \ref{fig}. 
In comparison, the best known classical algorithm, 
by Yuster and Zwick \cite{Yuster+05}, has complexity 
$\tilde O(n^2)$ if  $m\le n^{1.151}$,
$\tilde O(m^{0.697}n^{1.199})$ if $n^{1.151}\le m\le n^{(1+\omega)/2}$,
and
$\tilde O(n^\omega)$ if $n^{(1+\omega)/2}\le m\le n^2$.
Our algorithm performs better when $n^{1.151}< m< n^{\omega-1/2}$. 
For instance, if $m=O(n^{(1+\omega)/2})=O(n^{1.686...})$, 
then our algorithm has complexity
$O(n^{2.277})$, while the algorithm by \cite{Yuster+05} has complexity $\tilde O(n^\omega)$.

\begin{figure}\label{fig}
\begin{center}
 \includegraphics[scale=0.8]{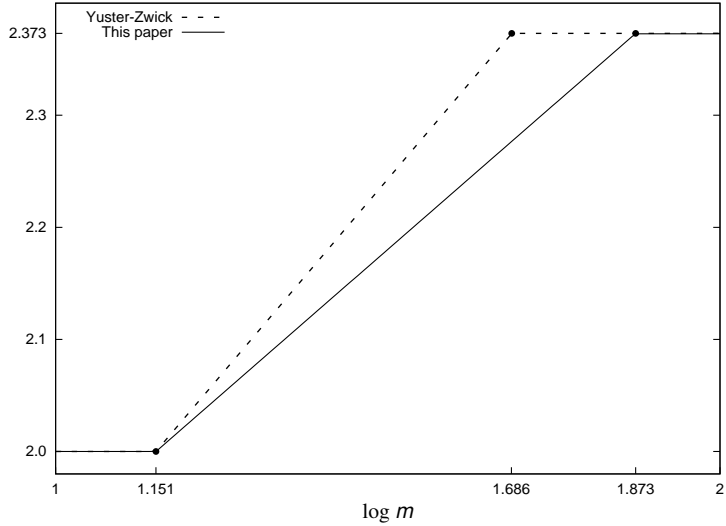}\vspace{-3mm}
\caption{The upper bounds of Theorem \ref{th_Booleansparse}  (in solid lines). 
The horizontal axis represents the logarithm of~$m$ with respect to basis $n$ (i.e., the value $\log_n(m)$). 
The vertical axis represents the logarithm of the complexity with respect to basis $n$.
The dashed lines represent the upper bounds obtained by \cite{Yuster+05}.
}
\end{center}\vspace{-8mm}
\end{figure}

Our main quantum tool is rather standard: quantum enumeration, a variant of Grover's search algorithm. 
We use this technique in various ways to improve the combinatorial steps in several
classical approaches \cite{Amossen+09,Duan+SODA09,Vassilevska+09,Yuster+05} 
that are based on a combination of algebraic steps (computing some matrix products over a field) 
and combinatorial steps. 
Moreover, the speedup obtained by quantum enumeration enables us to depart from these original approaches and 
optimize the combinatorial and algebraic steps in different ways, for instance relying on rectangular matrix multiplication
instead of square matrix multiplication. 
On the other hand, several subtle but crucial issues appear when trying to apply quantum enumeration, 
such as how to store and access information computed during the preprocessing steps, which
induces complications and requires the introduction of new algorithmic ideas.
We end up with algorithms fairly remote from these original
approaches, where most steps are tailored for the use of quantum enumeration. 
More detailed technical overviews are given at the 
beginning of Sections  \ref{sec_existence}, \ref{sec_maxmin} and~\ref{sec_red}. 

\section{Preliminaries}\label{prelim}
\noindent{\bf Rectangular matrix multiplication over fields}.
For any $k_1,k_2,k_3>0$, let $\omega(k_1,k_2,k_3)$ represent the minimal value $\tau$
such that, over a field, the product of an $n^{k_1}\times n^{k_2}$ matrix by an 
$n^{k_2}\times n^{k_3}$ matrix can be computed with $\tilde O(n^\tau)$ arithmetic operations.
The value $\omega(1,1,1)$ is denoted by $\omega$,
and the current best upper bound on $\omega$ is $\omega<2.373$, see \cite{Stothers10,WilliamsSTOC12}.
Other important quantities are 
the value $\alpha=\sup\{k\:|\:\omega(1,k,1)=2\}$ and the value
$\beta=(\omega-2)/(1-\alpha)$. 
The current best lower bound on $\alpha$ is $\alpha>0.302$, see \cite{LeGallFOCS12}.
The following facts are known, and will be used in this paper.
We refer to $\cite{Burgisser+97,Huang+98}$ for details.

\begin{fact}\label{fact_RMM1}
$\omega(1,k,1)=2$ for  $k\le \alpha$ and $\omega(1,k,1)\le 2+\beta(k-\alpha)$ for $\alpha\le k\le 1$.
\end{fact}

\begin{fact}\label{fact_RMM2}
The following relations hold for any values $k_1,k_2,k_3>0$: 
(i) $\omega(kk_1,kk_2,kk_3)=k\omega(k_1,k_2,k_3)$ for any $k>0$; 
(ii) $\omega(k_{\pi(1)},k_{\pi(2)},k_{\pi(3)})=\omega(k_1,k_2,k_3)$ for any permutation $\pi$ over $\{1,2,3\}$; 
(iii) $\omega(k_1,k_2,1+k_3)\le \omega(k_1,k_2,1)+k_3$; 
(iv) $\omega(k_1,k_2,k_3)\ge \max\{k_1+k_2,k_1+k_3,k_2+k_3\}$.
\end{fact}

\noindent{\bf Matrix products over semirings.}
We define below two matrix products over semirings considered
in Sections \ref{sec_existence} and \ref{sec_maxmin}, respectively,
additionally to the Boolean product, the $(\max,\min)$-product and the distance product defined in the introduction.
These products were also used in \cite{Duan+SODA09,Vassilevska+STOC06,Vassilevska+09}.
\begin{definition}
Let $A$ be an $n\times n$ matrix with entries in $\Int\cup\{\infty\}$
and $B$ be an $n\times n$ matrix with entries in $\Int\cup\{-\infty\}$.
The existence dominance product of $A$ and $B$, denoted $A\ast B$, is the $n\times n$ Boolean
matrix $C$ such that $C[i,j]=1$ if and only if there exists some $k\in\{1,\ldots,n\}$
such that $A[i,k]\le B[k,j]$. 
The product $A\smallerprod B$ is the $n\times n$ 
matrix $C$ such that 
$C[i,j]=-\infty$ if $A[i,k]> B[k,j]$ for all $k\in\{1,\ldots,n\}$, and 
$C[i,j]=\max_k\{A[i,k]\:|\:A[i,k]\le B[k,j]\}$ otherwise.
\end{definition}

It is easy to check, as mentioned for instance in \cite{Duan+SODA09,Vassilevska+09},
that computing the $(\max,\min)$-product reduces to computing the product $\smallerprod$.
Indeed if $C$ denotes the $(\max,\min)$-product of two matrices $A$ and $B$, then for any 
$(i,j)\in\{1,\ldots,n\}\times\{1,\ldots,n\}$ we can write
$
C[i,j]=\max\left\{(A\smallerprod B)[i,j],(B^T\smallerprod A^T)[j,i]\right\},
$ 
where $A^T$ and $B^T$ denote the transposes of $A$ and $B$, respectively.
Matrix products over the semirings $(\min,\max)$, $(\min,\le)$ and $(\max,\ge)$ 
studied, for instance, in \cite{VassilevskaPhD08}, similarly reduce to 
computing the product $\smallerprod$.\vspace{2mm}

\noindent{\bf Quantum algorithms for matrix multiplication.}
We assume that a quantum algorithm can access 
any entry of the input matrix in a random access way, similarly to the standard 
model used
in \cite{Buhrman+SODA06,Jeffery+ICALP12,LeGallSODA12,LeGallISAAC12}
for Boolean matrix multiplication. More precisely,
let $A$ and $B$ be two $n\times n$ matrices, for any positive integer $n$
(the model presented below can be 
generalized to deal with rectangular matrices in a straightforward way).
We suppose that these matrices can be accessed directly by a quantum algorithm:
We have an oracle~$O_A$ that, for any $i,j\in\{1,\ldots,n\}$, 
and any  $z\in\{0,1\}^\ast$, maps the state
$\ket{i}\ket{j}\ket{0}\ket{z}$ to $\ket{i}\ket{j}\ket{A[i,j]}\ket{z}$.
We have a similar oracle $O_B$ for $B$. 
Since we are interested in time complexity, we will count all the computational steps of the algorithm and assign a cost of one for each call 
to $O_A$ or $O_B$, which corresponds to the cases where quantum access to the inputs $A$ and $B$ can be done at 
unit cost, for example in a random access model working in quantum superposition.
We say that
a quantum algorithm for matrix multiplication computes the product of $A$ and $B$
with high probability if, 
when given access to oracles $O_A$ and $O_B$ corresponding to
$A$ and $B$, the algorithm outputs 
with probability at least $2/3$ all the entries of the product of $A$ and $B$.
The complexity of several algorithms in this paper will be stated using an upper bound $\lambda$ on the number 
of non-zero or non-infinite entries in the product of $A$ and $B$.  The same complexity, up to a logarithmic factor, can actually be obtained 
even if no nontrivial upper bound is known a priori, see \cite{LeGallSODA12,LeGallISAAC12}.

We will use variants of Grover's search algorithm, as described for instance in~\cite{Boyer+98},
to find elements satisfying some conditions inside a search space of size~$N$. Concretely, suppose that 
a Boolean function $f\colon\{1,\ldots,N\}\to\{0,1\}$ is given and that we want to find a solution, i.e., an element $x\in\{1,\ldots,n\}$ such that $f(x)=1$.
Consider the quantum search procedure (called safe Grover search in \cite{Magniez+SICOMP07}) obtained by repeating Grover's standard search 
a logarithmic number of times, and checking if a solution has been found. 
This quantum procedure outputs one solution
with probability at least $1-1/\poly(N)$ if a solution exists, 
and always rejects if no solution exists. Its time complexity
is $\tilde O(\sqrt{N/\max(1,t))})$, where $t$ denotes the number
of solutions, if the function $f$ can be evaluated in $\tilde O(1)$ time.
By repeating this procedure and striking out solutions as soon as 
they are found, one can find all the solutions with probability at least 
$1-1/\poly(N)$ using $\tilde O\big(\sqrt{N/t}+\sqrt{N/(t-1)}+\cdots+\sqrt{N/1}\big) =\tilde O(\sqrt{N(t+1)})$ computational steps.
We call this procedure \emph{quantum enumeration}.

\section{Existence Dominance Matrix Multiplication}\label{sec_existence}

In this section we present a quantum algorithm that computes the existence dominance 
product of two matrices $A$ and $B$. The underlying idea of our algorithm is similar 
to the idea in the best classical algorithm for the same problem by Duan and Pettie 
\cite{Duan+SODA09}: use a search step to find some of the entries of $A\domprod B$, 
and rely on classical algebraic algorithms to find the other entries. We naturally use
quantum search to implement the first part, and perform careful modifications of
their approach to improve the complexity in the quantum setting, taking advantage 
of the features of quantum enumeration.
There are two notable differences:  The first one is that the algebraic part of our quantum algorithms uses rectangular matrix multiplication, while \cite{Duan+SODA09} uses square matrix multiplication. 
The second and crucial difference is that, for applications in later sections, we give a quantum algorithm that can handle a more general version of the existence dominance product, defined on set of matrices, which we call the {\em generalized existence dominance product} and define below.


\begin{definition}\label{def:gene}
Let $u,v$ be two positive integers, and $S$ be the set $S=\{1,\ldots,u\}\times \{1,\ldots,v\}$.
Let $\prec$ be the lexicographic order over $S\cup\{(0,0)\}$ (i.e, $(i,j)\prec (i',j')$ if and only 
if $i<i'$ or ($i=i'$ and $j<j'$)).
Consider $u$ matrices
$A^{(1)},\ldots,A^{(u)}$, each of size $n\times n$
with entries in $\Int\cup\{\infty\}$, 
and $v$ matrices $B^{(1)},\ldots,B^{(v)}$, each of size $n\times n$
with entries in $\Int\cup\{-\infty\}$. 
For each $(i,j)\in\{1,\ldots,n\}\times\{1,\ldots,n\}$ define the set $S_{ij}\subseteq S\cup\{(0,0)\}$
as follows:
\[
S_{ij}=\{(x,y)\in S \:|\: A^{(x)}\ast B^{(y)}[i,j]=1\}\cup\{(0,0)\}.
\] 
The generalized existence dominance product of these matrices is the 
$n\times n$ matrix $C$ with entries in $S\cup\{(0,0)\}$ 
defined as follows: for all $(i,j)\in\{1,\ldots,n\}\times\{1,\ldots,n\}$ the entry 
$C[i,j]$ is the maximum element in $S_{ij}$,
where the maximum refers to the lexicographic order.
\end{definition}
Note that the case $u=v=1$ corresponds to the standard existence dominance product,
since $C[i,j]=(1,1)$ if $A^{(1)}\ast B^{(1)}[i,j]=1$ and $C[i,j]=(0,0)$ if $A^{(1)}\ast B^{(1)}[i,j]=0$. 

\begin{proposition}\label{prop_densedom}
Let $A^{(1)},\ldots,A^{(u)}$ be $u$ matrices of size $n\times n$
with entries in $\Int\cup\{\infty\}$, 
and $B^{(1)},\ldots,B^{(v)}$ be $v$ matrices of size $n\times n$
with entries in $\Int\cup\{-\infty\}$.
Let $m_1\in\{1,\ldots,n^2u\}$ denote the total number of finite entries in 
the matrices $A^{(1)},\ldots,A^{(u)}$, 
and $m_2\in\{1,\ldots,n^2v\}$ denote the total number of finite entries in 
the matrices $B^{(1)},\ldots,B^{(v)}$. 
For any parameter $t\in\{1,\ldots, m_1\}$, there exists a quantum algorithm
that computes, with high probability, 
their generalized existence dominance product in time
\[
\tilde O\left(\sqrt{\frac{m_1m_2n}{t}}+\sqrt{\frac{m_1m_2uv}{tn}}+n^{\omega(1+\log_n u,1+\log_n t,1+\log_n v)}\right).
\]
\end{proposition}
\begin{proof}
Let $t\in\{1,\ldots,m_1\}$ be a parameter to be chosen later.
Let $L$ be the list of all finite entries in $A^{(1)},\ldots,A^{(u)}$ sorted in increasing order.
Decompose $L$ into $t$ successive parts $L_1,\ldots,L_t$, each containing at most $\ceil{m_1/t}$
entries. For each $x\in\{1,\ldots,u\}$ and each
$r\in\{1,\ldots,t\}$ we construct two $n\times n$ matrices $A^{(x)}_r, \bar{A}^{(x)}_r$
as follows: for all $(i,j)\in\{1,\ldots,n\}\times \{1,\ldots,n\}$,
\begin{align*}
A^{(x)}_r[i,j]=&\left\{
\begin{array}{ll}
A^{(x)}[i,j] \textrm{ if }A^{(x)}[i,j]\in L_r,\\
\infty \textrm{ otherwise, }
\end{array}
\right.
&\hspace{6mm}
\bar{A}^{(x)}_r[i,j]=&\left\{
\begin{array}{ll}
1 \textrm{ if }A^{(x)}[i,j]\in L_r,\\
0 \textrm{ otherwise. }
\end{array}
\right.
\end{align*}
Similarly, for each $y\in\{1,\ldots,v\}$ and each
$r\in\{1,\ldots,t\}$ we construct two $n\times n$ matrices 
$B^{(y)}_r, \bar{B}^{(y)}_r$
as follows: for all $(i,j)\in\{1,\ldots,n\}\times \{1,\ldots,n\}$,
\begin{align*}
B^{(y)}_r[i,j]=&\left\{
\begin{array}{ll}
B^{(y)}[i,j] \textrm{ if }\min L_r\le B^{(y)}[i,j]< \max L_r,\\
-\infty \textrm{ otherwise, }
\end{array}
\right.
&\hspace{6mm}
\bar{B}^{(y)}_r[i,j]=&\left\{
\begin{array}{ll}
1 \textrm{ if }B^{(y)}[i,j]\ge \max L_r,\\
0 \textrm{ otherwise. }
\end{array}
\right.
\end{align*}
The cost of this (classical) preprocessing step is $O(n^2t(u+v))$ time.

It is easy to see that, for each $x\in\{1,\ldots,u\}$ and $y\in\{1,\ldots,v\}$, 
the following equality holds (where the operators $+$  and $\sum$ refer to the entry-wise OR):
\begin{equation}\label{eq_rel}
A^{(x)}\ast B^{(y)}=\sum_{r=1}^t\left(\bar{A}^{(x)}_r\Boolprod\bar{B}_r^{(y)}\right)+\sum_{r=1}^t\left(A^{(x)}_r\domprod B^{(y)}_r\right).
\end{equation}
Indeed, the second term compares entries that are in a same part $L_r$, while the first term takes into consideration entries in distinct parts.
Define two $n\times n$ matrices $C_1$ and $C_2$ with entries in $S\cup\{(0,0)\}$ as follows: for all $(i,j)\in\{1,\ldots,n\}\times \{1,\ldots,n\}$,
\begin{align*}
C_1[i,j]=&\max\left\{\{(0,0)\}\cup \{(x,y)\in S\:|\: \sum_{r=1}^t\bar A_r^{(x)}\Boolprod \bar B_r^{(y)}[i,j]=1\}\right\},\\
C_2[i,j]=&\max\left\{\{(0,0)\}\cup \{(x,y)\in S\:|\: \sum_{r=1}^t A_r^{(x)}\ast  B_r^{(y)}[i,j]=1\}\right\}.
\end{align*}
From Equation (\ref{eq_rel}), the generalized existence dominance product $C$ satisfies
$
C[i,j]=\max\{C_1[i,j],C_2[i,j]\}
$
for all $(i,j)\in\{1,\ldots,n\}\times \{1,\ldots,n\}$.
The matrix $C$ can then be computed in time $O(n^2)$ from $C_1$ and $C_2$.

The matrix $C_1$ can clearly be computed in time $O(n^2uv)$ if all
the terms $\sum_{r}\bar A_r^{(x)}\Boolprod \bar B_r^{(y)}$ are known.
We can obtain all these $uv$ terms by computing
the following Boolean product of an $nu\times nt$ matrix 
by an $nt\times nv$ matrix (both matrices can be constructed in time
$\tilde O(n^2t(u+v))$).
\[
\left[
\begin{array}{ccc}
\bar A^{(1)}_1&\cdots&\bar A^{(1)}_t\\
\vdots&&\vdots\\
\bar A^{(u)}_1&\cdots&\bar A^{(u)}_t
\end{array}
\right]
\cdot
\left[
\begin{array}{ccc}
\bar B^{(1)}_1&\cdots\cdots&\bar B^{(v)}_1\\
\vdots&&\vdots\\
\bar B^{(1)}_t&\cdots\cdots&\bar B^{(v)}_t\\
\end{array}
\right]
\]
The cost of this matrix multiplication is 
$\tilde O\left(n^{\omega(1+\log_n u,1+\log_n t,1+\log_n v)}\right)$.
From item (iv) of Fact \ref{fact_RMM2}, we conclude that
the matrix $C_1$ can be computed in time
\[
\tilde O\left(n^2uv+n^2t(u+v)+n^{\omega(1+\log_n u,1+\log_n t,1+\log_n v)}\right)=
\tilde O\left(n^{\omega(1+\log_n u,1+\log_n t,1+\log_n v)}\right).
\]

We use the following lemma to help us compute the matrix $C_2$.
While this is the main technical part of the proof of this proposition,
for readability its proof is placed in the appendix.
\begin{lemma}\label{lemma_dom}
There exists a quantum algorithm
that, with high probability, outputs 
\begin{itemize}
\item
$tu$ Boolean matrices $\hat A_r^{(x)}$, each of size $n\times 2n$, for all $x\in\{1,\ldots,u\}$ and $r\in\{1,\ldots,t\}$, 
\item
$tv$ Boolean matrices $\hat B_r^{(y)}$, each of size $2n\times n$, for all $y\in\{1,\ldots,v\}$ and $r\in\{1,\ldots,t\}$, 
\item
a matrix $D$ of size $n\times n$ with entries in $S\cup\{(0,0)\}=(\{1,\ldots,u\}\times \{1,\ldots,v\})\cup\{(0,0)\}$,
\end{itemize}
such that 
\[
C_2[i,j]=\max\left\{\{D[i,j]\}\cup \{(x,y)\in S\:|\: \sum_{r=1}^t\hat A_r^{(x)}\Boolprod \hat B_r^{(y)}[i,j]=1\}\right\}
\]
for all $(i,j)\in\{1,\ldots,n\}\times \{1,\ldots,n\}$.
The time complexity of this quantum algorithm is
\[
\tilde 
O\left(
n^2t(u+v)+\sqrt{\frac{m_1m_2n}{t}}+\sqrt{\frac{m_1m_2uv}{tn}}
\right).
\]
\end{lemma}
After applying the quantum algorithm of Lemma \ref{lemma_dom}, 
we can obtain the matrix $C_2$, similarly
to the computation of $C_1$, if we know 
all the terms $\sum_r\hat A_r^{(x)}\Boolprod \hat B_r^{(y)}$.
We obtain all these $uv$ terms by computing 
the following Boolean product of an $nu\times nt$ matrix 
by an $nt\times nv$ matrix.
\[
\left[
\begin{array}{ccc}
\hat A^{(1)}_1&\cdots&\hat A^{(1)}_t\\
\vdots&&\vdots\\
\hat A^{(u)}_1&\cdots&\hat A^{(u)}_t
\end{array}
\right]
\cdot
\left[
\begin{array}{ccc}
\hat B^{(1)}_1&\cdots\cdots&\hat B^{(v)}_1\\
\vdots&&\vdots\\
\hat B^{(1)}_t&\cdots\cdots&\hat B^{(v)}_t\\
\end{array}
\right]
\]
The cost of this matrix multiplication is 
$\tilde O\left(n^{\omega(1+\log_n u,1+\log_n t,1+\log_n v)}\right)$.
The total cost of computing the matrix $C_2$ is thus
\[
\tilde O\left(
n^2t(u+v)+\sqrt{\frac{m_1m_2n}{t}}+\sqrt{\frac{m_1m_2uv}{tn}}
+n^{\omega(1+\log_n u,1+\log_n t,1+\log_n v)}\right),
\]
which is the desired bound since the term 
$n^2t(u+v)$ is negligible here
by item (iv) of Fact \ref{fact_RMM2}.
\end{proof}

We can give a classical version of this result, whose proof can be found in the appendix,
that will be used to prove Theorem \ref{th_dist_cl} in Subsection \ref{subsec_dist}.
\begin{proposition}\label{prop_densedom_cl}
There exists a classical algorithm that computes the generalized existence dominance product 
in time
$
\tilde 
O\left(
\frac{m_1m_2}{tn}
+n^{\omega(1+\log_n u,1+\log_n t,1+\log_n v)}
\right),
$
for any parameter $t\in\{1,\ldots, m_1\}$.
\end{proposition}

We now consider the case $u=v=1$ corresponding to the standard existence dominance product. 
By optimizing the choice of the parameter $t$ in Proposition \ref{prop_densedom}, 
we obtain the following theorem.

\begin{theorem}\label{th_densedom}
Let  $A$ be an $n\times n$ matrix with entries in $\Int\cup\{\infty\}$
containing at most $m_1$ non-($\infty$) entries, and 
$B$ be an $n\times n$ matrix with entries in $\Int\cup\{-\infty\}$
containing at most $m_2$ non-($-\infty$) entries. 
There exists a quantum algorithm that computes, with high probability,
the existence dominance product of $A$ and $B$ in time
$
\tilde O(\sqrt{m_1m_2n^{1-\mu}}),
$
where $\mu$ is the solution of the equation
$
\mu+2\omega(1,1+\mu,1)=1+\log_n(m_1m_2).
$
In particular, this time complexity is upper bounded by 
$
\tilde O\left((m_1m_2)^{1/3}n^{(\omega+1)/3}\right).
$
\end{theorem}
\begin{proof}
The complexity of the algorithm of Proposition \ref{prop_densedom}
is minimized for $t=n^\mu$, where $\mu$ is the solution of the equation
$
\mu+2\omega(1,1+\mu,1)=1+\log_n(m_1m_2).
$
We can use items (ii) and (iii) of Fact \ref{fact_RMM2} to obtain the upper bound
$\omega(1,1+\mu,1)\le \omega+\mu$, and optimize
the complexity of the algorithm by taking
$t=\ceil{(m_1m_2)^{1/3}n^{(1-2\omega)/3}}$, which gives the upper bound claimed in the second part of the theorem.
\end{proof}

In the case of completely dense input matrices (i.e., $m_1\approx n^2$ and $m_2\approx n^2$), 
the second part of Theorem~\ref{th_densedom} shows that the complexity of the algorithm is 
$\tilde O(n^{(5+\omega)/3})\le O(n^{2.458})$. 

\section{Applications: the $(\max,\min)$-Product and the Distance Product}\label{sec_maxmin}
In this section we show how to apply the results of Section \ref{sec_existence}
to construct quantum algorithms for the $(\max,\min)$-product and the distance product.
\subsection{Quantum Algorithm for the $(\max,\min)$-Product}\label{subsec_maxmin}
In this subsection we present a quantum algorithm for the matrix product $\smallerprod$,
which immediately gives a quantum algorithm with the same complexity for the 
$(\max,\min)$-product as explained in Section \ref{prelim}, and then gives Theorem~\ref{the2}.
Our algorithm first exploits the methodology by Vassilevska et al.~\cite{Vassilevska+09} 
to reduce the computation of the product $\smallerprod$ to the computation of 
several sparse dominance products.
The main technical difficulty to 
overcome is that, unlike in the classical case, computing all the sparse dominance products 
successively becomes too costly (i.e., the cost exceeds the complexity of all the other parts 
of the quantum algorithm). 
Instead, we show that it is sufficient to obtain a small fraction of the entries in each dominance product
and that this task 
reduces to the computation of a generalized existence dominance product, and then
use the quantum techniques of Proposition \ref{prop_densedom}  to 
obtain precisely only those entries.

\begin{theorem}
There exists a quantum algorithm that computes,
for any two $n\times n$ matrices $A$ and $B$
with entries respectively in $\Int\cup\{\infty\}$ and $\Int\cup\{-\infty\}$,
the product $A\smallerprod B$ with high probability in time 
$
\tilde O(n^{(5-\gamma)/2}),
$
where $\gamma$ is the solution of the equation
$
\gamma+2\omega(1+\gamma,1+\gamma,1)=5.
$
In particular, this complexity is upper bounded by 
$
O(n^{2.473}).
$
\end{theorem}
\begin{proof}
Let $g\in\{1,\ldots,n\}$ be a parameter to be chosen later.
For each $i\in\{1,\ldots,n\}$, we sort the entries in the $i$-th
row of $A$ in increasing order and divide the list into  
$s=\ceil{n/g}$ successive parts $R_{1}^i,\ldots,R_s^i$ 
with at most $g$ entries in each part.
For each $r\in\{1,\ldots, s\}$, define the $n\times n$ matrix $A_r$
as follows: $A_r[i,j]=A[i,j]$ if $A[i,j]\in R_r^i$ and $A_r[i,j]=\infty$
otherwise.
The cost of this (classical) preprocessing is $O(n^2s)$ time. 

We describe below the quantum algorithm that computes $C=A\smallerprod B$.\vspace{2mm}

\noindent {\bf Step 1.} 
For each $(i,j)\in\{1,\ldots,n\}\times \{1,\ldots,n\}$, we compute the largest $r\in\{1,\ldots,s\}$ such that 
$(A_r\domprod B)[i,j]=1$, if such an $r$ exists.
This is done by using the quantum algorithm
of Proposition \ref{prop_densedom} with $u=s$, $v=1$, $A^{(r)}=A_r$ for each $r\in\{1,\ldots, s\}$ 
and $B^{(1)}=B$. Note that $m_1\le s\times (ng)=O(n^2)$ and $m_2\le n^2$. The complexity of this step 
is thus
\[
\tilde 
O\left(
\frac{n^{5/2}}{\sqrt{t}}+n^{\omega(1+\log_ns,1+\log_nt,1)}
\right)
\]
for any parameter $t\in\{1,\ldots,n^2\}$.
We want to minimize this expression.
Let us write $t=n^\gamma$ and $g=n^\delta$. For a fixed $\delta$, the first term 
is a decreasing function of $\gamma$, while the second term is an increasing function 
of $\gamma$. The expression is thus minimized for the value of $\gamma$ solution of the 
equation
\begin{equation}\label{eqdelta}
\omega(2-\delta,1+\gamma,1)=(5-\gamma)/2,
\end{equation}
in which case the expression becomes $\tilde O(n^{(5-\gamma)/2})$.
\vspace{2mm}

\noindent {\bf Step 2.} 
Note that at Step 1 we also obtain all $(i,j)\in\{1,\ldots,n\}\times \{1,\ldots,n\}$ such that no~$r$ satisfying $(A_r\ast B)[i,j]=1$
exists. For all those $(i,j)$, we set $C[i,j]=-\infty$.
For all other $(i,j)$, we will denote by $r_{ij}$ the value found at Step 1.
We now know that 
\[
C[i,j]=\max_{k :\: A[i,k]\in R_{r_{ij}}^i}\{A_{r_{ij}}[i,k]\:|\:A_{r_{ij}}[i,k]\le B[k,j]\},
\]
and $C[i,j]$ can be computed in time 
$\tilde O(\sqrt{g})$
using the quantum algorithm for maximum finding \cite{Durr+96}, since $|R_{r_{ij}}^i|\le g$. The complexity of Step 2 is
thus $\tilde O(n^2\sqrt{g})$.\vspace{2mm}

This algorithm computes, with high probability, all the entries of $C=A\smallerprod B$. Its complexity is 
\[
\tilde O\left(n^2s+n^{(5-\gamma)/2}+n^2\sqrt{g} \right)=\tilde O\left(n^{(5-\gamma)/2}+n^{2+\delta/2} \right),
\]
since the term 
$n^2s=n^{3-\delta}$ is negligible 
with respect to $n^{(5-\gamma)/2}=n^{\omega(2-\delta,1+\gamma,1)}$
by item (iv) of Fact \ref{fact_RMM2}. 
This expression is minimized for $\delta$ and $\gamma$ satisfying $\delta+\gamma=1$. Injecting this 
constraint into
Equation (\ref{eqdelta}), we find that the optimal value of $\gamma$ is the solution of the equation
$
\gamma+2\omega(1+\gamma,1+\gamma,1)=5,
$
as claimed. Using items (i) and (ii) of Fact \ref{fact_RMM2} and Fact \ref{fact_RMM1},
we obtain 
\[
5= \gamma+2(1+\gamma)\omega\!\left(\!\!1,1,\frac{1}{1+\gamma}\right) \le
\gamma+2(1+\gamma)\left(2+\beta\left(\frac{1}{1+\gamma}-\alpha\right)\!\!\right)=
(4+2\beta-2\alpha\beta)+(5-2\alpha\beta)\gamma
\]
and then
$
\gamma\ge\frac{1+2\alpha\beta-2\beta}{5-2\alpha\beta}.
$
The complexity is thus 
$\tilde O\left(n^{(12-6\alpha\beta+\beta)/(5-2\alpha\beta)}\right)\le O(n^{2.473})$.
\end{proof}

\subsection{Quantum Algorithm for the  Distance Product}\label{subsec_dist}
In this subsection we present a quantum algorithm that computes the most significant bits of the
distance product of two matrices,
as defined below.

Let $A$ and $B$ be two $n\times n$ matrices with entries in $\Int\cup\{\infty\}$.
Let $W$ be a power of two such that the value of each finite entry of their distance product $C$ is upper bounded by~$W$.
For instance, one can take the smallest power of two larger than $\max_{i,j}\{A[i,j]\}+\max_{i,j}\{B[i,j]\}$,
where the maxima are over the finite entries of the matrices.
Each non-negative finite entry of $C$ can then be expressed using $\log_2(W)$ bits:
the entry $C[i,j]$
can be expressed as 
$
C[i,j]=\sum_{k=1}^{\log_2(W)}C[i,j]_k \frac{W}{2^k}
$
for bits $C[i,j]_1,\ldots,C[i,j]_{\log_2(W)}$.
For any $\ell\in\{1,\ldots,\log_2(W)\}$, we say that an algorithm computes the $\ell$
most significant bits of each entry if, for all $(i,j)\in\{1,\ldots,n\}\times \{1,\ldots,n\}$ such 
that $C[i,j]$ is finite and non-negative, 
the algorithm outputs all the bits $C[i,j]_1,C[i,j]_2,\cdots ,C[i,j]_{\ell}$.
Vassilevska and Williams~\cite{Vassilevska+STOC06} have studied this problem, and 
shown how to reduce the computation
of the $\ell$
most significant bits to the computation of $O(2^\ell)$ existence dominance matrix products
of $n\times n$ matrices. 
By combining this with the $\tilde O(n^{(3+\omega)/2})$-time algorithm for 
dominance product from \cite{Matousek91},
they obtained
a classical algorithm that computes
the 
$\ell$
most significant bits of each entry of 
the distance product of $A$ and $B$
in time
$
\tilde O\left (2^{\ell} n^{(3+\omega)/2}\right)\le \tilde O\left (2^{\ell} n^{2.687}\right).
$

%

Here is the main result of this subsection, 
obtained by reducing 
the computation
of the $\ell$
most significant bits to computing a 
generalized existence dominance product.
\begin{theorem}\label{th_dist1}
There exists a quantum algorithm that computes,
for any two $n\times n$ matrices $A$ and $B$ with entries in $\Int\cup\{\infty\}$,
the 
$\ell$
most significant bits of each entry of 
the distance product of $A$ and $B$
in time
$
\tilde O\left(2^{0.640\ell} n^{(5+\omega)/3}\right)\le O(2^{0.640\ell} n^{2.458})
$
with high probability.
\end{theorem}
\begin{proof}
Note that the trivial $\tilde O(n^{5/2})$-time quantum algorithm can be used to compute 
all the bits of each entry of the distance product $C$ of $A$ and $B$. Therefore, we will assume,
without loss of generality, that $\ell$ satisfies the inequality $2^{0.640\ell} n^{(5+\omega)/3}\le n^{5/2}$,
which implies in particular that $2^\ell\le n^2$.

Assume first that all the entries of $C$ are finite and non-negative.
What we want to do is to compute, for each $(i,j)\in\{1,\ldots,n\}\times \{1,\ldots,n\}$,
the integer
$d\in\{0,1\ldots, 2^\ell-1\}$
such that $C[i,j]$ is in the interval $[dW/2^\ell,(d+1)W/2^\ell)$. 

For any integer $x$, define the matrices $A'_x$ and $B'_x$ as follows:
for all $(i,j)\in\{1,\ldots,n\}\times \{1,\ldots,n\}$,
\begin{align*}
A'_x[i,j]&=A[i,j]-\frac{xW}{2^{\ell/2}},
\hspace{20mm}B'_x[i,j]=-B[i,j]+\frac{xW}{2^{\ell}}.
\end{align*}

Assume for simplicity that $\ell$ is even (a similar argument works for $\ell$ odd).
For each $d\in\{0,1\ldots, 2^\ell-1\}$, let $d_1,d_2\in\{0,1\ldots, 2^{\ell/2}-1\}$ denote the values  
such that $d=d_12^{\ell/2}+d_2$.
For each $d\in\{0,1\ldots, 2^\ell-1\}$, define the Boolean matrix 
$
D_{d}=A'_{d_1}\ast B'_{d_2},
$ 
where $\ast$ means the strict\footnote{The strict existence dominance product is obtained by replacing $\le$ by $<$ in the definition of the existence dominance product (Definition 2.1). Note that all our results on the existence dominance product also hold for the strict existence dominance product and their proofs are essentially the same, just replacing inequalities by strict inequalities.} existence dominance product. 
Note that 
$
\frac{d_1W}{2^{\ell/2}}+\frac{d_2W}{2^{\ell}}=\frac{dW}{2^{\ell}}.
$
Observe that, for each $(i,j)\in\{1,\ldots,n\}\times \{1,\ldots,n\}$, we have 
\[
D_d[i,j]=0 \Longleftrightarrow\min_{k}\left(A[i,k]+B[k,j]\right)\ge \frac{dW}{2^{\ell}}.
\]

For each $(i,j)\in\{1,\ldots,n\}\times \{1,\ldots,n\}$,  
the integer $d\in\{1,\ldots, 2^\ell\}$
such that $C[i,j]$ is in the interval $[(d-1)W/2^\ell,dW/2^\ell)$
can thus be found by computing the smallest $d\in\{0,1\ldots, 2^\ell-1\}$
such that $D_d[i,j]=1$. We can thus use\footnote{Actually, we need to modify the
order $\prec$ 
in Definition \ref{def:gene} so that the algorithm of
Proposition \ref{prop_densedom}
finds the smallest $d$ such that $D_d[i,j]=1$ instead of the largest $d$.
This is done simply by choosing $\prec$ as the decreasing lexicographic order
instead of the usual lexicographic order.
Proposition \ref{prop_densedom} and its proof are unchanged, since the proof
only uses the fact that $\prec$ is a strict total order.} 
the quantum algorithm of 
Proposition \ref{prop_densedom}, with $u=v=2^{\ell/2}$, 
$A^{(x)}=A'_{x-1}$ and $B^{(y)}=B'_{y-1}$ for each $x,y\in\{1,\ldots,2^{\ell/2}\}$.
Since $m_1\le 2^{\ell/2}n^2$, $m_2\le 2^{\ell/2}n^2$ and from the inequality $2^\ell\le n^2$ on $\ell$, the complexity is 
\[
\tilde 
O\left(
\frac{n^{5/2}2^{\ell/2}}{\sqrt{t}}
+n^{\omega(1+\log_n(2^{\ell/2}),1+\log_n t,1+\log_n(2^{\ell/2}))}
\right)
\]
for any parameter $t\in\{1,\ldots,2^{\ell/2}n^2\}$.
Let us write $\mu=\log_n(2^\ell)$ and $t=n^\gamma$. 
The complexity is minimized for the 
value $\gamma$ such that
$
5+\mu-\gamma=2\omega(1+\mu/2,1+\gamma,1+\mu/2),
$
for which the complexity is $\tilde O\left(n^{(5+\mu-\gamma)/2}\right)$.
Using items (i) and (iii) of Fact \ref{fact_RMM2} and Fact \ref{fact_RMM1}, we obtain
\[
\omega(1+\mu/2,1+\gamma,1+\mu/2)\le \gamma+(1+\mu/2)\omega\left(\!\!1,\frac{1}{1+\mu/2},1\right)\le \gamma+(1+\mu/2)\left(2+\beta\left(\frac{2}{2+\mu}-\alpha\right)\right).
\]
This gives
$
5+\mu-\gamma\le 2\gamma+(2-\alpha\beta)\mu+(4+2\beta- 2\alpha\beta)
$
and thus 
$
\gamma\ge \frac{(\alpha\beta-1)\mu+(1-2\beta+2\alpha\beta)}{3}.
$
The complexity is thus
$
\tilde O\left(
n^{\frac{5}{2}+\frac{2\beta-2\alpha\beta-1}{6}+\frac{(4-\alpha\beta)}{6}\mu}
\right)=
\tilde O\left(
n^{\frac{5+\omega}{3}+0.640\mu}
\right)
=
O\left(
2^{0.640\ell}n^{2.458}
\right).
$

Finally, we discuss the general case
where the entries of $C$ can be negative or infinite.
Observe that the above algorithm 
detects which entries of $C$ are larger than
$(2^\ell-1)W/2^\ell$: 
these are the entries such that the algorithm 
finds no $d$ such that $D_d[i,j]=1$.
We can find which of these entries are larger than $W$
(and thus infinite) by computing the dominance product
$
A'_{0}\ast B'_{2^\ell}.
$ 
Note that the algorithm also finds which entries of $C$ are negative: 
these are the entries for which the 
smallest $d$ such that $D_d[i,j]=1$ is $d=0$. 
\end{proof}


Similarly, we can obtain a better classical algorithm as shown in the following theorem.
\begin{theorem}\label{th_dist_cl}
There exists a classical algorithm that computes,
for any two $n\times n$ matrices $A$ and $B$ with entries in $\Int\cup\{\infty\}$,
the 
$\ell$
most significant bits of each entry of 
the distance product of $A$ and $B$
in time
$
\tilde O\big( 2^{0.960\ell}n^{(3+\omega)/2}\big)\le O(2^{0.960\ell} n^{2.687}).
$
\end{theorem}
\begin{proof}
The proof is similar to the proof of Theorem \ref{th_dist1}, but we use Proposition \ref{prop_densedom_cl} instead of Proposition \ref{prop_densedom}.
The complexity becomes 
\[
\tilde 
O\left(
\frac{2^{\ell}n^{3}}{t}
+n^{\omega(1+\log_n(2^{\ell/2}),1+\log_n t,1+\log_n(2^{\ell/2}))}
\right)
\]
for any parameter $t\in\{1,\ldots,2^{\ell/2}n^2\}$.
Let us write $\mu=\log_n(2^\ell)$ and $t=n^\gamma$. 
This expression is then  
\[
O\left(
n^{3+\mu-\gamma}
+n^{\omega(1+\mu/2,1+\gamma,1+\mu/2)}
\right).
\]
This expression is minimized for the 
value $\gamma$ such that
\[
3+\mu-\gamma=\omega(1+\mu/2,1+\gamma,1+\mu/2),
\]
for which the complexity is $\tilde O\left(n^{3+\mu-\gamma}\right)$.

Using items (i) and (iii) of Fact \ref{fact_RMM2} and Fact \ref{fact_RMM1}, we obtain
\[
\omega(1+\mu/2,1+\gamma,1+\mu/2)\le \gamma+(1+\mu/2)\omega\left(1,\frac{1}{1+\mu/2},1\right)
\le\gamma+(1+\mu/2)\left(2+\beta\left(\frac{2}{2+\mu}-\alpha\right)\right).
\]
This gives the inequality
\[
3+\mu-\gamma\le \gamma+(1-\frac{\alpha\beta}{2})\mu+(2+\beta- \alpha\beta),
\]
from which we obtain 
\[
\gamma\ge \frac{\alpha\beta\mu/2+(1-\beta+\alpha\beta)}{2}.
\]
The complexity is thus 
\[
\tilde O\left(
n^{\frac{(5+\beta-\alpha\beta)}{2}+(1-\frac{\alpha\beta}{4})\mu}
\right)=
\tilde O\left(
n^{\frac{3+\omega}{2}+0.960\mu}
\right)
=
O\left(
2^{0.960\ell} n^{2.687}
\right),
\]
as claimed.
\end{proof}

Note that the dependency on $n$ of the $\tilde O\left (2^{\ell} n^{2.687}\right)$-time
algorithm by Vassilevska and Williams~\cite{Vassilevska+STOC06}
can be slightly improved using the recent
$O(n^{2.684})$-time algorithm for 
dominance product by Yuster~\cite{YusterSODA09} based on rectangular matrix multiplication.
We can similarly obtain an improved bound $O(2^{c\ell} n^{2.684})$, for some $c<1$, 
with the same approach as in the proof of Theorem \ref{th_dist_cl}. However, it is complicated to express
the value of $c$ in a closed form, so we omit the statement of this slight improvement.


\section{Sparse Boolean Matrix Multiplication}\label{sec_red}
In this section we describe quantum versions of several known combinatorial techniques for handling sparse
Boolean matrix products. 
The main result is the following theorem, which shows how to compute the Boolean product of two matrices $A$ and $B$ 
by reducing it to four products, each easier to compute than the original one when $A$ and $B$ are sparse enough. 
Note that similar ideas have been used in \cite{Amossen+09,Yuster+05} to analyze applications of those combinatorial techniques in the classical setting. Here we show how to implement these ideas using quantum enumeration and analyze the complexity of the resulting algorithm.
\begin{theorem}\label{prop}
Assume that there exists an algorithm that computes, 
in time $M(n_1,n_2,n_3,L)$,
the product of any $n_1\times n_2$ Boolean matrix and 
any $n_2\times n_3$ Boolean matrix such that their product  
contains at most $L$ non-zero entries. 
Let $A$ and $B$ be two $n\times n$ Boolean matrices with at most $m_1$ and $m_2$ non-zero entries in $A$ and~$B$,
respectively. For any values of the three parameters $\ell_1\in\{1,\ldots,m_1\}$ and $\ell_2,\ell_3\in\{1,\ldots,m_2\}$,
there exists a quantum algorithm that computes, with high probability,
the Boolean product $A\Boolprod B$ and has time complexity
\[
\tilde O\Bigg(M(\ell'_1,\ell'_2,\ell'_3,\lambda)+\sqrt{\frac{m_1m_2\cdot \min(\lambda,m_1m_2/\ell_2)}{\ell_2}}+\lambda\sqrt{\frac{m_1}{\ell_1}}+
\lambda\sqrt{\frac{m_2}{\ell_3}}+n^2\Bigg),
\] 
where $\lambda$ denotes the number of non-zero entries in $A\Boolprod B$, and $\ell'_i=\min(\ell_i,n)$ for each $i\in\{1,2,3\}$.
\end{theorem}
\begin{proof}
For any $k\!\in\!\{1,\ldots,n\}$, 
let $a^R_k$ (resp.~$b^R_k$) be the number of non-zero entries in the $k$-th row of $A$ (resp.~$B$)
and
$a^C_k$ (resp.~$b^C_k$) be the number of non-zero entries in the $k$-th column of $A$ (resp.~$B$).
We define the following six sets of indexes, and compute them classically in time $O(n^2)$.
\begin{align*}
S=&\left\{k\in\{1,\ldots,n\}\:|\:b^R_k\ge m_2/\ell_2\right\}&\hspace{5mm}
S'=&\left\{k\in\{1,\ldots,n\}\:|\:b^R_k< m_2/\ell_2\right\}\\
T=&\left\{k\in\{1,\ldots,n\}\:|\:a^R_k\ge m_1/\ell_1\right\}&
T'=&\left\{k\in\{1,\ldots,n\}\:|\:a^R_k< m_1/\ell_1\right\}\\
U=&\left\{k\in\{1,\ldots,n\}\:|\:b^C_k\ge m_2/\ell_3\right\}&
U'=&\left\{k\in\{1,\ldots,n\}\:|\:b^C_k< m_2/\ell_3\right\}
\end{align*}

Given two sets $R,C\subseteq\{1,\ldots,n\}$ and an $n\times n$ Boolean matrix $M$, the notation
$M_R^C$ will represent the $n\times n$ Boolean matrix such that $M^C_R[i,j]=1$ if and only if $M[i,j]=1$ and $(i,j)\in R\times C$.
For convenience, $M_R$ will represent the matrix $M_R^C$ for $C=\{1,\ldots, n\}$,
and $M^C$ the matrix $M_R^C$ for $R=\{1,\ldots, n\}$.

It is easy to check that
\[
A\Boolprod B=A_T^S\Boolprod B_S^U+A^{S'}\Boolprod B_{S'}+A_{T'}\Boolprod B +A\Boolprod B^{U'},
\]
where $+$ represents the entry-wise OR operation. We will individually compute the four terms of this sum.

The computation of $A_T^S\Boolprod B_S^U$ consists in the computation of a $|T|\times |S|$ matrix by a $|S|\times |U|$ matrix.
We implement this part using the algorithm whose existence is assumed in the statement of the theorem.
Note that, from the 
sparsity of $B$, we have 
$
m_2\ge \sum_{k=1}^n b^R_k\ge \sum_{k\in S} b^R_k\ge |S|m_2/ \ell_2,
$
and thus $|S|\le\ell_2$. Similarly we have $|T|\le \ell_1$ and $|U|\le \ell_3$.
Additionally, we know that $S$, $T$ and $U$ have size at most $n$.
Thus this part can be implemented in
$
M(\ell'_1,\ell'_2,\ell'_3,\lambda)
$
time.

In order to compute $A^{S'}\Boolprod B_{S'}$ we do the following. 
Let $m'_1$ denote the number of non-zero entries of $A^{S'}$.
First, we list all these non-zero entries, classically in time $O(n^2)$, and 
record them into two arrays $\arrayf{M_1}$ and $\arrayf{M_2}$ of size $m'_1$:
for each $p\in\{1,\ldots,m'_1\}$ the value $\arrayf{M_1}[p]$ records the row index of the 
$p$-th element of the list, while $\arrayf{M_2}[p]$ records its column index. 
Then, for each $k\in S'$, we compute the set of indexes $j\in\{1,\ldots,n\}$ such that 
$B[k,j]=1$ and record them into an array $\arrayf{N_k}$. Note that $\arrayf{N_k}$ 
has length $b_k^R$, and that $b_k^R<m_2/\ell_2$ from the definition of $S'$. 
The computation of all the $\arrayf{N_k}$'s can be done classically in $O(n^2)$ time. 
Finally, take $N=\sum_{c=1}^{m'_1} b^R_{\arrayf{M_2}[c]}$ and define the 
function $g\colon\{1,\ldots,N\}\to \{1,\ldots,n\}\times \{1,\ldots,n\}$ as follows:
for any $p\in\{1,\ldots,m'_1\}$ and any $q\in \{1,\ldots,b^R_{\arrayf{M_2}[p]}\}$,
\[
g\left(q+\sum_{c=1}^{p-1}b^R_{\arrayf{M_2}[c]}\right)=(\arrayf{M_1}[p],\arrayf{N}_{\arrayf{M_2}[p]}[q]),
\]
where $\arrayf{N}_{\arrayf{M_2}[p]}[q]$ denotes the $q$-th element of the array $\arrayf{N}_{\arrayf{M_2}[p]}$.
It is easy to check that
\[
g(\{1,\ldots,N\})=
\left\{(i,j)\in\{1,\ldots,n\}\times \{1,\ldots,n\}\:|\: \textrm{there exists $k\in S'$ such that }A[i,k]=B[k,j]=1\right\},
\]
i.e., $g(\{1,\ldots,N\})$ is precisely the set of non-zero entries of $A^{S'}\Boolprod B_{S'}$ that we want to find.
A crucial point 
here is that the function $g$ can be evaluated in $\poly(\log n)$
time using the data structures $\arrayf{M_1}$, $\arrayf{M_2}$ and $\arrayf{N_k}$.
For any subset $\Sigma$ of $\{1,\ldots,n\}\times \{1,\ldots,n\}$,
let $f_\Sigma\colon\{1,\ldots,N\}\to\{0,1\}$ be the function such that $f_\Sigma(x)=1$ if
and only if $g(x)\notin \Sigma$.
The quantum procedure starts with $\Sigma$ being empty,
performs successive quantum searches over $\{1,\ldots,N\}$, 
each time searching for an element $x$ such that $f_\Sigma(x)=1$
and adding $g(x)$ to $\Sigma$ as soon as such an $x$ is found,
and stops when no new element $x$ is found. 
From the discussion of Section \ref{prelim}, with high probability 
all searches succeed, in which case at the end of the procedure
$\Sigma=g(\{1,\ldots,N\})$.
Let $\lambda'$ denote the number of non-zero entries in $A^{S'}\Boolprod B_{S'}$ and observe that
$\lambda' \le \min(\lambda,m_1m_2/\ell_2)$, since $N<m_1'm_2/\ell_2\le m_1m_2/\ell_2$.
The overall complexity of this quantum procedure is
\begin{equation*}
\tilde O\left(n^2+\sqrt{N\times (\lambda'+1)}\right)=
\tilde O\left(n^2+\sqrt{\frac{m_1m_2\cdot\min(\lambda,m_1m_2/\ell_2)}{\ell_2}}\right).
\end{equation*}


The computation of $A_{T'}\Boolprod B$ is done as follows. For each $k\in\{1,\ldots, n\}$, let~$a'_k$ denote the number of non-zero entries in the $k$-th column of $A_{T'}$. 
We first perform a $O(n^2)$-time classical preprocessing step: for each $k\in\{1,\ldots,n\}$, we construct the set 
$E_k$ of the row indexes of all non-zero entries in the $k$-th column of $A_{T'}$,
and construct the set $F_k$ of the column indexes of all non-zero entries in the $k$-th row of $B$.
Note that $|E_k|=a'_k$ and $|F_k|=b^R_k$.
The quantum procedure computing 
$A_{T'}\Boolprod B$
uses a set $\Sigma\subseteq \{1,\ldots,n\}\times \{1,\ldots,n\}$, initially empty.
For each $k\in\{1,\ldots,n\}$,
all the $(i,j)\in E_k\times F_k$ such that $A_{T'}[i,k]=B[k,j]=1$ and $(i,j)\notin \Sigma$ are computed by 
performing a quantum enumeration, as above, over
%
the set $E_k\times F_k$, 
adding $(i,j)$ to $\Sigma$ as soon as such a $(i,j)$ is found,
and stopping when no new element $(i,j)$ is found.
The overall time complexity is   
$
\tilde O\left(n^2+\sum_{k=1}^n \sqrt{a'_kb^R_k(\lambda_k+1)}\right),
$
where $\lambda_k$ is the number of elements found 
when processing $k$. 
Note that the inequality 
$
\sum_{k} a'_kb^R_k< \lambda \min(m_1/\ell_1,n)
$
holds,
since $\sum_{k} a'_kb^R_k$ also represents the total number of witnesses of $A_{T'}\Boolprod B$,
i.e., the number of triples $(i,j,k)$ such that $A_{T'}[i,k]=B[k,j]=1$ (observe that there are at most
$\lambda$ pairs $(i,j)$ satisfying this condition, all such that $i\in T'$).
Since $\sum_{k}\lambda_k\le\lambda\le n^2$,
this complexity is upper bounded by
\[
\tilde O\left(n^2+\sqrt{\lambda+n}\times \sqrt{\sum_{k=1}^n a'_kb^R_k}\right)
=
\tilde O\left(n^2+\sqrt{\lambda +n}\times 
\sqrt{\lambda\min\left(\frac{m_1}{\ell_1},n\right)}
\right)
=
\tilde O\left(n^2+\lambda\times \sqrt{\frac{m_1}{\ell_1}}\right).
\]

Computing $A\Boolprod B^{U'}$ is done similarly to the computation of $A_{T'}\Boolprod B$
with cost
$
\tilde O\left(n^2+\lambda\sqrt{m_2/\ell_3}\right).
$
\end{proof}

We now compare the results of Theorem \ref{prop} to previous works. 
For the case $m_1=m_2\approx n^2$, the bounds obtained in Theorem \ref{prop} are not better than the best known output-sensitive 
algorithms for Boolean matrix multiplication \cite{Jeffery+ICALP12,LeGallISAAC12,Lingas11}. Interestingly, we nevertheless recover
the same complexity $O(\lambda\sqrt{n})$ as in~\cite{LeGallISAAC12} for the region $n^{3/2}\le \lambda\le n^{2}$, but using different methods
(this is done by taking $\ell_1=m_1/(n+1)$, which gives $T'=\{1,\ldots,n\}$ and
reduces the computation of $A\Boolprod B$ to the computation of only $A_{T'}\Boolprod B$).
Consider now sparse input matrices and, for concreteness, focus on
the case $m_1=m_2$  (we denote this value simply by $m$). The complexity of the 
algorithm by Amossen and Pagh \cite{Amossen+09}, while not stated in this form, can be written as
\[
\tilde O\left(M(\ell'_1,\ell'_2,\ell'_1,\lambda)+m^2/\ell_2+\lambda m/\ell_1+n^2\right)
\] 
using the notations of Theorem \ref{prop}. In comparison, 
Theorem \ref{prop} gives (by choosing $\ell_1=\ell_3$) the upper bound
\[
\tilde O\left(M(\ell'_1,\ell'_2,\ell'_1,\lambda)+\min(m\sqrt{\lambda/\ell_2},m^2/\ell_2)+\lambda\sqrt{m/\ell_1}+n^2\right).
\] 
We see that the second and third terms in our complexity are never worse. 
In order to evaluate quantitatively the speedup obtained in the quantum setting, 
let us consider the case when only the input matrices are sparse (i.e., $\lambda\approx n^2$).
In this case, the algorithm by Amossen and Pagh has the same complexity as the algorithm by 
Yuster and Zwick \cite{Yuster+05} described in the introduction. In comparison, 
Theorem~\ref{prop} gives the following result, which shows that our quantum 
algorithm is better than their classical algorithm, as discussed in the introduction.
\addtocounter{section}{-4}
\begin{theorem}[complete version]
Let $A$ and $B$ be two $n\times n$ Boolean matrices with at most $m_1$ and $m_2$ non-zero entries in $A$ and $B$,
respectively.
There exists a quantum algorithm that computes, with high probability,
the Boolean matrix product $A\Boolprod B$ and has time complexity
\[
\left\{\begin{array}{ll}
\tilde O(n\times \min(m_1,m_2))&\textrm{if  }1\le \sqrt{m_1m_2}\le n,\\
\tilde O(n^2)&\textrm{if }n\le \sqrt{m_1m_2}\le n^{1+\alpha/2},\\
\tilde O\left((m_1m_2)^\frac{\beta}{1+2\beta}n^\frac{2+2\beta-\alpha\beta}{1+2\beta}\right)&
\textrm{if }n^{1+\alpha/2}\le \sqrt{m_1m_2}\le n^{\omega-1/2},\\
\tilde O(n^\omega)&\textrm{if }n^{\omega-1/2}\le \sqrt{m_1m_2}\le n^2.\\
\end{array}
\right. 
\] 
\end{theorem}
\addtocounter{section}{+4}
\begin{proof}
First consider the case $\sqrt{m_1m_2}\le n$.
Assume for now 
that $m_1\le m_2$. We
use the following strategy: we first use quantum enumeration to find
all the non-zero entries of $A$ and, then, for each such entry $A[i,k]$, we output all the 
$j$'s such that $B[k,j]=1$. The complexity of this strategy is $\tilde O(\sqrt{(m_1+1)n^2}+m_1n)=\tilde O(m_1n)$.
The same argument gives the upper bound $\tilde O(m_2n)$ when $m_2\le m_1$.

If $n\le \sqrt{m_1m_2}\le n^{1+\alpha/2}$, then we use the quantum algorithm of Theorem \ref{prop} with parameters $\ell_1=m_1$,
$\ell_2=m_1m_2/n^2$, $\ell_3=m_2$, and applying the algorithm for rectangular matrix multiplication over a field described in 
Section \ref{prelim} for the part $M(\ell'_1,\ell'_2,\ell'_3,n^2)$. This gives overall complexity $\tilde O(n^2)$ time.

If $n^{1+\alpha/2}\le \sqrt{m_1m_2}\le n^{\omega-1/2}$, then we use the quantum algorithm of Theorem~\ref{prop} with parameters
$\ell_1=m_1$, $\ell_3=m_2$ and 
\[
\ell_2=(m_1m_2)^\frac{1}{1+2\beta}n^{\frac{2(\alpha\beta-1)}{1+2\beta}},
\]
giving overall complexity 
\[
\tilde O\left((m_1m_2)^\frac{\beta}{1+2\beta}n^\frac{2+2\beta-\alpha\beta}{1+2\beta}\right).
\]

Finally, if $\sqrt{m_1m_2}\ge n^{\omega-1/2}$, then we simply use the best existing  classical algorithm for dense matrix multiplication.
\end{proof}


\section*{Appendix: Proofs of Lemma \ref{lemma_dom} and Proposition \ref{prop_densedom_cl}}
In this appendix we give the proofs of Lemma \ref{lemma_dom} and Proposition \ref{prop_densedom_cl}.

\begin{proof}[Proof of Lemma \ref{lemma_dom}]
In the proof we will use the notation $\col{M}{k}$ to denote 
the number of finite entries in the $k$-th row of $M$, for any 
$n_1\times n_2$ matrix $M$ with entries in $\Int\cup\{\pm \infty\}$
and any $k\in\{1,\ldots,n_2\}$.

Our algorithm proceeds in several steps.\vspace{2mm}

\noindent{\bf Preprocessing: column balancing}

For each $r\in\{1,\ldots,t\}$, we do the following.
Consider the $u$ matrices $A^{(1)}_r,\ldots,A_r^{(u)}$. Each matrix has size $n\times n$ 
and  we know that the total number of finite entries in these $u$ matrices is at most $\ceil{m_1/t}$:
\begin{equation}\label{proof_eq1}
\sum_{x=1}^u\sum_{k=1}^n \col{A_r^{(x)}}{k}\le \ceil{m_1/t}.
\end{equation}
We will construct $u$ matrices $\tilde A^{(1)}_r, \ldots, \tilde A^{(u)}_r$, each of size $n\times 2n$.
Each $\tilde A^{(x)}_r$ will contain all the finite entries in $A_r^{(x)}$, but these $u$ matrices will
satisfy the following sparsity condition on each column: 
\begin{equation}\label{proof_eq2}
\sum_{x=1}^u\col{\tilde A_r^{(x)}}{k'}\le \ceil{m_1/(nt)} \:\:\:\textrm{ for all $k'\in\{1,\ldots,2n\}$}.
\end{equation}
These matrices are related to
the concept of column balancing developed in \cite{Duan+SODA09}.

Let us describe how to construct these matrices $\tilde A^{(1)}_r, \ldots, \tilde A^{(u)}_r$.
For each $k\in\{1,\ldots,n\}$, we first collect together all the finite entries in the $k$-th column of $A_r^{(1)},\ldots,A_r^{(u)}$
and sort them in increasing order. This gives, for each $k$, a sorted list of at most $nu$ numbers, with possible repetitions.
We then divide this list into successive parts $T_{r,k}^1,T_{r,k}^2,\ldots,T_{r,k}^{a_{r,k}}$, for some $a_{r,k}\ge 1$, such that 
\begin{equation*}\label{proof_eq2b}
\left\{
\begin{array}{ll}
|T_{r,k}^q|=\ceil{m_1/(nt)} &\textrm{ for } q\in\{1,\ldots,a_{r,k}-1\},\\
|T_{r,k}^q|\le\ceil{m_1/(nt)}&\textrm{ for } q=a_{r,k}.
\end{array}
\right.
\end{equation*}
Define $p_r=\sum_{k=1}^n a_{r,k}$ and notice that $p_r\le 2n$:
there are at most $n$ parts of size exactly $\ceil{m_1/(nt)}$ due to Equation (\ref{proof_eq1}),
and at most $n$ parts of size strictly less than $\ceil{m_1/(nt)}$ (these parts are among the $n$
parts with $q=a_{r,k}$).
To each pair $(k,q)$ with $k\in\{1,\ldots,n\}$ and $q\in\{1,\ldots,a_{r,k}\}$,
we assign an arbitrary index in $\{1,\ldots,p_r\}$, denoted $\rho_r(k,q)$, in a bijective way. 
Finally, for each $x\in\{1,\ldots,u\}$, we construct the $n\times 2n$  matrix $\tilde A^{(x)}_r$ as follows: 
for all $i\in\{1,\ldots,n\}$ and all $k'\in\{1,\ldots,2n\}$,
\[
\tilde A^{(x)}_r[i,k']=
\left\{\begin{array}{ll}
A^{(x)}_r[i,k]&\textrm{ if } k'\in\{1,\ldots,p_r\} \textrm{ and  $A^{(x)}_r[i,k]\in T_{r,k}^{q}$, where $(k,q)=\rho_r^{-1}(k')$},\\
\infty&\textrm{ otherwise.}
\end{array}
\right. 
\]
This means that each finite entry of $A_r^{(x)}$ appears in $\tilde A_r^{(x)}$, in the same row but generally in a different column.
By construction, Equation (\ref{proof_eq2}) holds.
The overall cost of this (classical) preprocessing step is $\tilde O(n^2tu)$ time.\vspace{2mm}

\noindent{\bf Preprocessing: recording relevant information about the input matrices}

Since the complexity of the quantum procedure described in the last part of the proof will 
depend crucially on the way information about matrices $\tilde A_r^{(x)}$ and $B_r^{(y)}$ is stored,
we introduce adequate data structures to record this information.  

For each $x\in\{1,\ldots,u\}$, we do the following.
For all $r\in\{1,\ldots,t\}$ we list the finite entries in each column of $\tilde A_r^{(x)}$, classically in time $\tilde O(n^2t)$,
and create a 3-dimensional array $\arrayf{U}^{(x)}$ such that $\arrayf{U}^{(x)}[r,k',b]$ records the index of the row of the $b$-th finite 
entry in the $k'$-th column of $\tilde A_r^{(x)}$, for each $r\in\{1,\ldots,t\}$, each $k'\in\{1,\ldots,2n\}$,
and each $b\in\{1,\ldots,\col{\tilde A_r^{(x)}}{k'}\}$. 

For each $y\in\{1,\ldots,v\}$, we do the following. 
We construct, classically in time $O(n^2t)$, a list 
containing all the finite entries of $B_1^{(y)},\ldots,B_t^{(y)}$.
Let us denote the total number of these finite entries by  $m^{(y)}$, and 
remember that we have $\sum_{y=1}^v m^{(y)}\le m_2$. We then create an array 
$\arrayf{V}^{(y)}$ of size $m^{(y)}$:  for each $a\in\{1,\ldots,m^{(y)}\}$, if the $a$-th element of the list is
$B^{(y)}_r[k,j]$, then $\arrayf{V}^{(y)}[a]$ is set to the 3-tuple $(r,k,j)$. 

The overall cost of this (classical) preprocessing step is $\tilde O(n^2t(u+v))$ time.
\vspace{2mm}

\noindent{\bf Construction of the matrices $\hat A_r^{(x)}$ and $\hat B_r^{(y)}$}

For each $r\in\{1,\ldots,t\}$ and each $x\in\{1,\ldots,u\}$, we construct an $n\times 2n$ Boolean matrix $\hat A^{(x)}_r$
as follows: for all $i\in\{1,\ldots,n\}$
and all $k'\in\{1,\ldots,2n\}$,
\[
\hat A^{(x)}_r[i,k']=1 \textrm{ iff }\tilde A^{(x)}_r[i,k']\neq \infty .
\]
For each $r\in\{1,\ldots,t\}$ and each $y\in\{1,\ldots,v\}$, we construct an $2n\times n$ Boolean matrix $\hat B_r^{(y)}$
as follows: for all $k'\in\{1,\ldots,2n\}$ and all $j\in\{1,\ldots,n\}$,
\[
\hat B_r^{(y)}[k',j]=1 \textrm{ iff } k'\in\{1,\ldots,p_r\} \textrm{ and } B_r^{(y)}[k,j]\ge \max T_{r,k}^{q}, \textrm{ where }(k,q)=\rho_r^{-1}(k').
\]
These are the matrices mentioned in the statement of the lemma.
The overall cost of this (classical) construction step is $\tilde O(n^2t(u+v))$ time.\vspace{2mm}

\noindent{\bf Relation with the matrix $C_2$}

For each $r\in\{1,\ldots,t\}$, each $x\in\{1,\ldots,u\}$ and each $y\in\{1,\ldots,v\}$,
consider the Boolean product $\hat A_r^{(x)}\Boolprod\hat B_r^{(y)}$.
This product gives us some of the non-zero entries of $A_r^{(x)}\ast B_r^{(y)}$,
but not all.
Indeed, by definition, $\hat A_r^{(x)}[i,k']=1$ if and only if $A^{(x)}_r[i,k]\in T_{r,k}^{q}$,
where $(k,q)=\rho_{r}^{-1}(k')$.
The indexes of the non-zero entries of $\hat A_r^{(x)}\Boolprod\hat B_r^{(y)}$ are thus precisely all the 
$(i,j)\in\{1,\ldots,n\}\times \{1,\ldots,n\}$ for which there exists some $k\in\{1,\ldots,n\}$ satisfying 
\[
A^{(x)}_r[i,k]\in T_{r,k}^{q} \textrm{ for some } q\in\{1,\ldots,a_{r,k}\} \textrm{ and } B_r^{(y)}[k,j]\ge \max T_{r,k}^{q}.
\]

Let us now consider the remaining non-zero entries of 
$A_r^{(x)}\ast B_r^{(y)}$: the $(i,j)\in\{1,\ldots,n\}\times \{1,\ldots,n\}$
for which there exists some $k\in\{1,\ldots,n\}$ satisfying
\begin{equation}\label{eq_cond}
A_r^{(x)}[i,k]\in T_{r,k}^{q} \textrm{ for some } q\in\{1,\ldots,a_{r,k}\} \textrm{ and } A_r^{(x)}[i,k]\le B_r^{(y)}[k,j]< \max T_{r,k}^{q}.
\end{equation}
Define the $n\times n$ matrix $D$ with entries in $S\cup \{(0,0)\}$ as follows.
For any $(i,j)\in\{1,\ldots,n\}\times \{1,\ldots,n\}$, the entry $D[i,j]$ is the largest
element $(x,y)\in S$ such that Equation (\ref{eq_cond}) holds for some $r\in\{1,\ldots,t\}$
and some $k\in\{1,\ldots,n\}$, if at least one such $(x,y)$ exists, and $D[i,j]=(0,0)$
otherwise.

We then have
\[
C_2[i,j]=\max\left\{\{D[i,j]\}\cup \{(x,y)\in S\:|\: \sum_{r=1}^t\hat A_r^{(x)}\Boolprod \hat B_r^{(y)}[i,j]=1\}\right\},
\]
for all $(i,j)\in\{1,\ldots,n\}\times \{1,\ldots,n\}$,
as claimed in the statement of the lemma.

\vspace{2mm}

\noindent{\bf Construction of the matrix $D$}

We finally show how to compute the matrix $D$. The idea is to find,
for all $(x,y)\in S$ in decreasing order,
all the pairs of indexes $(i,j)\in\{1,\ldots,n\}\times \{1,\ldots,n\}$ such that Equation (\ref{eq_cond}) holds
for some $r\in\{1,\ldots,t\}$ and some $k\in\{1,\ldots,n\}$, and strike out those 
pairs as soon as they are found.

\begin{figure}
\fbox{
\begin{minipage}{16 cm}
\begin{codebox}
\li $R\gets\emptyset$;
\li \For all $(i,j)\in\{1,\ldots,n\}\times \{1,\ldots,n\}$ \kw{do}  $D[i,j]\gets(0,0)$; \kw{enddo} 
\li   \For all $y\in\{1,\ldots,v\}$, all $r\in\{1,\ldots t\}$ and all $(j,k)\in\{1,\ldots,n\}\times \{1,\ldots,n\}$ \kw{do}
\li   \hspace{5mm} compute the smallest $q\in\{1,\ldots,a_{r,k}\}$ satisfying $B^{(y)}_r[k,j]< \max T_{r,k}^{q}$ and denote it by $q^{(y)}_{rkj}$;
\li   \kw{enddo}
\li   \For all $(x,y)\in S$ in decreasing order \kw{do}
\li   \hspace{5mm} \id{test-full} $\gets$ \id{false};
\li   \hspace{5mm} \While \id{test-full} = \id{false} \kw{do}
\li   \hspace{10mm} find $(r,i,j,k)\in \{1,\ldots,t\}\times\{1,\ldots,n\}^3$ such that $(i,j)\not\in R$ and $\tilde A_r^{(x)}[i,\rho_r(k,q^{(y)}_{rkj})]\le B_r^{(y)}[k,j]$;
\zi  \hspace{-6mm}\# \hspace{13mm} \emph{comment: the search of Step 9 is actually done over } $\Gamma^{(x,y)}\subset \{1,\ldots,t\}\times\{1,\ldots,n\}^3$
\li   \hspace{10mm} \If a solution $(r,i,j,k)$ is found 
\li  \hspace{15mm}\kw{then} $D[i,j]\gets (x,y)$; $R\gets R\cup \{(i,j)\}$;
\li  \hspace{15mm}\kw{else} \id{test-full} $\gets \id{true}$;
\li   \hspace{5mm}\kw{enddo} 
\li   \hspace{0mm}\kw{enddo} 
\End
\end{codebox}\vspace{0mm}
\end{minipage}
}
\caption{Procedure computing the matrix $D$.}\label{fig:alg}
\end{figure}

The procedure for computing $D$ is described in Figure \ref{fig:alg}. The set $R$,
initially empty, records all pairs $(i,j)$ for which $D[i,j]$ has already been 
computed. During the loop of Steps 8-13 the procedure enumerates all the $(i,j)\in(\{1,\ldots,n\}\times \{1,\ldots,n\})\setminus R$
such that Equation (\ref{eq_cond}) holds for some $r\in\{1,\ldots,t\}$
and some $k\in\{1,\ldots,n\}$. Note that only the non-($-\infty$) entries of
$B_r^{(y)}$ need to be considered and, from Equation (\ref{eq_cond}),
for each such non-($-\infty$) entry $B_r^{(y)}[k,j]$ only
the non-($\infty$) entries $A_r^{(x)}[i,k]$ of $A_r^{(x)}$ such that  
\[
A_r^{(x)}[i,k]\in T_{r,k}^{q^{(y)}_{rkj}}
\]
need to be considered, 
where $q^{(y)}_{rkj}$ is the smallest integer in $\{1,\ldots,a_{r,k}\}$ such that 
\[
B^{(y)}_r[k,j]< \max T_{r,k}^{q^{(y)}_{rkj}}.
\] 
By construction, these non-($\infty$) entries of $A_r^{(x)}$ are in 
the $\rho_r(k,q^{(y)}_{rkj})$-th column of $\tilde A^{(x)}_r$.
The loop of Steps 8-13 thus
performs successive quantum searches over the set 
\[
\Gamma^{(x,y)}=\left\{
(r,i,j,k)\in \{1,\ldots,t\}\times\{1,\ldots,n\}^3\:|\: B_r^{(y)}[k,j]\neq -\infty \:\textrm{ and }\: \tilde A_r^{(x)}[i,\rho_r(k,q^{(y)}_{rkj})]\neq \infty
\right\},
\]
looking for elements $(r,i,j,k)\in \Gamma^{(x,y)}$ such that 
\[
(i,j)\notin R\: \textrm{ and }\: \tilde A_r^{(x)}[i,\rho_r(k,q^{(y)}_{rkj})]\le B_r^{(y)}[k,j].
\]

The procedure of Figure \ref{fig:alg} correctly computes the matrix $D$ whenever the quantum 
enumeration does not err, that is, with probability at least $1-1/\poly(n)$ if safe Grover search is
used, as discussed in Section \ref{prelim}. Let us consider its time complexity. 
The cost of Step 2 is $\tilde O(n^2)$, and the cost of the loop of Steps 3-5 is $\tilde O(n^2tv)$
since each $q^{(y)}_{rkj}$ can be found in $\poly(\log n)$ time using binary search. 

In order to evaluate the cost of the loop of Steps 8-13, we need to discuss in more details
how to perform the quantum search over the set $\Gamma^{(x,y)}$ since subtle issues 
arise when considering how to access time-efficiently the relevant entries of the matrices
and how to check if an element is a solution in $\poly(n)$ time.
Note that 
\[
|\Gamma^{(x,y)}|=
\sum_{r=1}^t
\sum_{\begin{subarray}{c}j,k \textrm{ such that }\\B_r^{(y)}[k,j]\neq -\infty\end{subarray}}
\col{\tilde A_r^{(x)}}{\rho_r(k,q_{rkj}^{(y)})}.
\]
We define a bijection $g$ from the set $\{1,\ldots,|\Gamma^{(x,y)}|\}$ to the set 
$\Gamma^{(x,y)}$ as follows. Remember that the data structures $\arrayf{U}^{(x)}$ and $\arrayf{V}^{(y)}$
 are available, recording information about the $\tilde A_r^{(x)}$'s and the $B_r^{(y)}$'s,
 respectively.
For notational convenience for each $a\in\{1,\ldots, m^{(y)}\}$ with corresponding value $\arrayf{V}^{(y)}[a]=(r,k,j)$, 
we will write $\arrayf{V}_1[a]=r$, $\arrayf{V}_2[a]=k$, $\arrayf{V}_3[a]=j$ and $\arrayf{W}[a]=\rho_r(k,q_{rkj}^{(y)})$.
Note that these four values can be immediately obtained from $\arrayf{V}^{(y)}[a]$.
We define the function $g$ 
as follows: for all $a\in\{1,\ldots,m^{(y)}\}$ and 
all $b\in\{1,\ldots,\col{\tilde A_{\arrayf{V_1}[a]}^{(x)}}{\arrayf{W}[a]}\}$,
\[
g\left(b+\sum_{c=1}^{a-1} \col{\tilde A_{\arrayf{V_1}[c]}^{(x)}}{\arrayf{W}[c]} \right)=(\arrayf{V_1}[a],\arrayf{U}^{(x)}\Big[\arrayf{V_1}[a],\arrayf{W}[a],b\Big],\arrayf{V_3}[a],\arrayf{V_2}[a]).
\]
It is easy to check that $g$ is a bijection from $\{1,\ldots,|\Gamma^{(x,y)}|\}$ to $\Gamma^{(x,y)}$.
The crucial point here is that the function $g$ can be evaluated in $\poly(\log n)$ time since 
$\arrayf{U}^{(x)}$ and $\arrayf{V}^{(y)}$ are available (in particular, given any $z\in \{1,\ldots,|\Gamma^{(x,y)}|\}$
one can find the values $a$ and $b$ such that 
$z=b+\sum_{c=1}^{a-1} \col{\tilde A_{\arrayf{V_1}[c]}^{(x)}}{\arrayf{W}[c]}$ efficiently,
using binary search for instance).
We can then implement Step 9 by performing quantum searches over the set 
$\{1,\ldots,|\Gamma^{(x,y)}|\}$. From the discussion in 
Section \ref{prelim}, the time complexity of the loop of Steps 8-13,
for fixed $(x,y)$,
is thus
\[
\tilde O\left(
\sqrt{|\Gamma^{(x,y)}|\times (\lambda^{(x,y)}+1)}\right)=
\tilde O\left(
\sqrt{\left(\sum_{r=1}^t \sum_{\begin{subarray}{c}j,k \textrm{ such that }\\B_r^{(y)}[k,j]\neq -\infty\end{subarray}}
\col{\tilde A_r^{(x)}}{\rho_r(k,q_{rkj}^{(y)})}\right)\times (\lambda^{(x,y)}+1)}\right),
\]
where $\lambda^{(x,y)}$ denotes the number of elements found during the execution of the loop 
(i.e., the number of new entries of $D$ computed).

The total cost of the procedure of Figure \ref{fig:alg} is then
\[
\tilde O\left(n^2tv+\sum_{x=1}^u\sum_{y=1}^v
\sqrt{\left(\sum_{r=1}^t \sum_{\begin{subarray}{c}j,k \textrm{ such that }\\B_r^{(y)}[k,j]\neq -\infty\end{subarray}}
\col{\tilde A_r^{(x)}}{\rho_r(k,q_{rkj}^{(y)})}\right)\times (\lambda^{(x,y)}+1)}\right)
\]
Using Equation (\ref{proof_eq2}), the inequality
$\sum_{x=1}^u\sum_{y=1}^v\lambda^{(x,y)}\le n^2$, and the Cauchy-Schwarz inequality, 
we can rewrite this expression as
\begin{align*}
\tilde O\left(n^2tv+\sum_{y=1}^v
\sqrt{\left(\sum_{r=1}^t \sum_{\begin{subarray}{c}j,k \textrm{ such that }\\B_r^{(y)}[k,j]\neq -\infty\end{subarray}}
\frac{m_1}{nt}\right)\times (u+\sum_{x=1}^u\lambda^{(x,y)})}\right)
&=\tilde O\left(n^2tv+
\sqrt{
\frac{m_1m_2}{nt}\times (uv+\sum_{x=1}^u\sum_{y=1}^v\lambda^{(x,y)})
}
\right)\\
&=\tilde O\left(n^2tv+\sqrt{\frac{m_1m_2n}{t}}+\sqrt{\frac{m_1m_2uv}{tn}}\right).
\end{align*}
This concludes the description of how to construct the matrix $D$.

Since the preprocessing has cost $\tilde O(n^2t(u+v))$, the overall complexity of 
the algorithm is 
\[
\tilde O\left(n^2t(u+v)+\sqrt{\frac{m_1m_2n}{t}}+\sqrt{\frac{m_1m_2uv}{tn}}\right).
\]
This concludes the 
proof of Lemma \ref{lemma_dom}.
\end{proof}

\begin{proof}[Proof of Proposition \ref{prop_densedom_cl}]
The algorithm is essentially the same as in the proof of Proposition \ref{prop_densedom}.
The only change is that a classical algorithm, which we describe below, is used
instead of the quantum algorithm in Lemma \ref{lemma_dom}.
In the proof of Lemma \ref{lemma_dom}
we use classical enumeration (i.e., exhaustive search) instead of quantum enumeration
in the procedure of Figure~\ref{fig:alg}.
The complexity of the loop of Steps 8-13 
of the procedure of Figure \ref{fig:alg}
is thus $|\Gamma^{(x,y)}|$, and
the total complexity of the procedure becomes
\begin{align*}
\tilde O\left(n^2tv+\sum_{x=1}^u\sum_{y=1}^v
\sum_{r=1}^t \sum_{\begin{subarray}{c}j,k \textrm{ such that }\\B_r^{(y)}[k,j]\neq -\infty\end{subarray}}
\col{\tilde A_r^{(x)}}{\rho_r(k,q_{rkj}^{(y)})}\right)
&=\tilde O\left(n^2tv+\frac{m_1m_2}{tn}\right).
\end{align*}
The overall complexity of 
the classical version of Lemma \ref{lemma_dom} is thus
\[
\tilde O\left(n^2t(u+v)+n^2tv+\frac{m_1m_2}{tn}\right)=\tilde O\left(n^2t(u+v)+\frac{m_1m_2}{tn}\right),
\]
which gives the upper bound claimed.
\end{proof}

\end{document}